\documentclass[11pt]{article}
\usepackage{amssymb}%
\usepackage{amsmath}%

\usepackage{graphicx,amsmath,amsfonts,amscd,amssymb,bm,cite,epsfig,epsf,url,color,algorithm,algorithmic}
\usepackage{fullpage} \usepackage[small,bf]{caption}
\newtheorem{theorem}{Theorem}[section]
\newtheorem{lemma}[theorem]{Lemma}
\newtheorem{corollary}[theorem]{Corollary}

\newcommand{\R}{\mathbb{R}}

\newcommand{\tr}{\operatorname{Tr}}

\numberwithin{equation}{section}

\def \endprf{\hfill {\vrule height6pt width6pt depth0pt}\medskip}
\newenvironment{proof}{\noindent {\bf Proof} }{\endprf\par}

\newcommand{\cS}{\mathbb{U}_n^4}

\newcommand{\St}{\mathbb{V}_{2}(\mathbb{C}^n)}

\begin{document}

\newcommand{\cA}{\mathcal{A}}
\newcommand{\xl}{x_{l}}
\newcommand{\xj}{x_{j}}
\newcommand{\bxl}{\bar{x}_{l}}
\newcommand{\bxj}{\bar{x}_{j}}
\newcommand{\zkl}{z_{l}}
\newcommand{\zkj}{z_{j}}
\newcommand{\bzkl}{\bar{z}_{l}}
\newcommand{\bzkj}{\bar{z}_{j}}
\newcommand{\zkabr}{z_{a}\bar{z}_{b}|z_{r}|^2}
\newcommand{\dlarm}{\delta_{(l-1-(a-r)=0,n)}}
\newcommand{\dlarp}{\delta_{(l-1+(a-r)=0,n)}}
\newcommand{\djbrm}{\delta_{(j-1-(b-r)=0,n)}}
\newcommand{\djbrp}{\delta_{(j-1+(b-r)=0,n)}}
\newcommand{\dbr}{\delta_{(b-r=0)}}
\newcommand{\dar}{\delta_{(a-r=0)}}

\newcommand{\dx}{\frac{\partial}{\partial x_i}}
\newcommand{\dy}{\frac{\partial}{\partial y_i}}
\newcommand{\dv}{\frac{\partial}{\partial v_i}}
\newcommand{\dw}{\frac{\partial}{\partial w_i}}

\newcommand{\cB}{\mathcal{B}}
\newcommand{\cU}{\mathbb{U}_n}
\newcommand{\cSU}{\mathbb{S}\mathbb{U}}

\newcommand{\Comp}{\mathbb{C}}

\title{Determination of all pure quantum states from a minimal number of observables.}

\author{ Damien Mondragon$\dagger$ and Vladislav Voroninski$\ddagger$ \\
  \vspace{-.1cm}\\
  $\dagger$ Department of Mathematics, University of California, Berkeley, CA 94720\\
  $\ddagger$ Department of Mathematics, Massachusetts Institute of Technology, MA 02139
}

\maketitle
\begin{abstract}

We show that for any positive integer $n$, the maps $x \in \mathbb{C}^n \mapsto \{\left|\langle x, z_i \rangle \right|^2\}_{i=1}^{4n} \in \mathbb{R}^{4n}$, where $z_i$ are the columns of four $n\times n$ unitary matrices, are generically injective modulo multiplication by a global phase factor, yielding a family of embeddings of $\mathbb{C}P^{n-1}$ into $\mathbb{R}^{4n-4}$. In particular, this implies that distribution measurements about a pure state with four generic full-rank observables are informationally complete, which is sharp for $n \geq 6$. To complement this information-theoretic study, we establish in a companion paper that the PhaseLift algorithm yields efficient phase retrieval from quadratic measurements with $O(1)$ unitary matrices, with high probability, where the unitaries are iid according to Haar measure.
\end{abstract}


%
%
%

\section{Introduction}
In 1933 Wolfgang Pauli posed what is now known as the "Pauli Problem" \cite{Pauli}: does perfect knowledge of the distributions obtained from the momentum and position observables uniquely determine any pure quantum state? It was answered in the negative \cite{PauliNegative}, and much research focused on characterizing ways in which the statement failed. Ron Wright conjectured in 1978 that there exist three observables which uniquely determine any pure state \cite{Vogt78}. An erratum of a paper by B.Z. Moroz in 1983 \cite{Moroz1,Moroz2} acknowledges M. Gromov for pointing out an argument that at least four observables are required in high enough dimensions, via general geometric obstructions to embedding $\mathbb{C}P^{n}$ into Euclidean space. In 1994, Moroz and Perelomov exposited this argument \cite{OPPP} and similar arguments were independently rediscovered more recently in \cite{QTPI}, being further elaborated upon in connection to Wright's conjecture in \cite{SFM,SavingPhase}. 


In this paper we settle the line of inquiry regarding the minimal number of informationally complete observables for dimensions $n \geq 6$: we show that distribution measurements with 4 full-rank observables are generically sufficient to determine any pure state $x \in \mathbb{C}P^{n-1}$. By the results of \cite{QTPI,SFM}, at least 4 observables are necessary for informational completeness when $n \geq 6$, thus our result is sharp in those dimensions. In particular, our work exhibits a family of embeddings of $\mathbb{C}P^{n-1}$ into $\mathbb{R}^{4(n-1)}$. From the perspective of phase retrieval, each element of this family yields $4n-3$ measurements vectors which are injective modulo multiplication by a global phase factor.

Informational completeness of observables on its own doesn't imply the existence of efficient state recovery algorithms. To complement our information theoretic study, we show in a companion paper \cite{QTFFRO} that efficient and exact phase retrieval is achievable in the same setting, with distribution measurements from a constant number of full-rank observables. Specifically, we show that the recently proposed PhaseLift algorithm \cite{CSV}, which consists of solving a simple semidefinite program, exactly recovers a fixed pure state from quadric measurements with vectors from $O(1)$ unitary matrices with high probability, provided the random matrices are iid according to Haar measure on the unitary group. Thus, the jump from the information theoretic recovery limit to guaranteed efficient exact phase retrieval, is a constant oversampling factor. 

\subsection{Notation and conventions}
We work in the standard finite dimensional setting of quantum mechanics. Each state $x \in \mathbb{C}^n$ is unit norm and defined up to a global phase factor: $x \sim e^{i\theta} x$ for any $\theta \in \mathbb{R}$. Let $\mathbb{H}_{n}$ be the set of Hermitian operators on $\mathbb{C}^n$ and denote the $n$-dimensional unitary group as $\mathbb{U}_n$.  An observable $A$ is then an element of $\mathbb{H}_{n }$, with eigenvalues $\lambda_i$ and associated eigenspaces $E_i$. Taking a measurement of a state $x$ with an observable yields $\lambda_i$ with probability $\| \mathcal{P}_{E_i}(x)\|_2^2$, where $\mathcal{P}_{E_i}$ is the projection on eigenspace $E_i$. Generically, an observable $A$ has $n$ distinct real eigenvalues $\lambda_1 > \lambda_2 \ldots > \lambda_n$ and in this case, the measurements are $\lambda_i$ with probability $\| u_i u_i^*(x)\|_2^2 = \left| \left< u_i,x\right>\right|^2$, where $u_1, \ldots u_n$ are an associated set of orthonormal eigenvectors of $A$. From now on, we say an observable $A \in \mathbb{H}_n$ is admissible, if it has $n$ distinct eigenvalues. We say that an element $U \in \mathbb{U}_n$ diagonalizes an observable $A$, if the columns of $U$ consist of eigenvectors of $A$ and $A = UDU^*$ where $D$ is a diagonal matrix such that $D_{ii} = \lambda_i$. 

We shall refer to a set of observables $\{A_{i}\}_{i=1}^m$ as informationally complete, if knowledge of the probabilities of observing their individual eigenvalues by measuring a particular pure state $x \in \mathbb{C}P^{n-1}$, uniquely determines any such pure state. That is, a set of observables $\{A_{i}\}_{i=1}^m$ is informationally complete if the map $x \in \mathbb{C}P^{n-1} \mapsto \left\{ \left|\left< u_j^{(i)}, x/\|x\|_2 \right> \right|^2 \right\}_{1\leq j \leq n, 1\leq i \leq m }$ is injective, where $u_{j}^{(i)}$ is the $j$'th column of $U_i$, and $A_i = U_i D U_i^*$ is a diagonalization of $A_i$. A slightly different convention is used in the field of phase retrieval, in which a set of measurement vectors $z_i \in \mathbb{C}, i = 1,2,\ldots,n $ is said to be injective modulo phase if quadratic measurements $\left|\left<z_i,x\right>\right|^2, i = 1,2,\ldots m$ determine any $x \in \mathbb{C}^n$ modulo multiplication by a complex number of unit norm, that is, as an element of $\mathbb{C}^{n}/\mathbb{S}^1$. We will say here that a matrix $A \in \mathbb{C}^{m \times n}$ is injective modulo phase if measurements $ \left\{| e_j^* Ax|^2 \right\}_{j=1}^m \equiv |Ax|^2 \in\mathbb{R}^m$, where $\{e_j\}_{j=1}^m$ is the standard basis for $\mathbb{R}^m$, determine any $x$ as an element of $\mathbb{C}^{n}/\mathbb{S}^1$. Note that in the standard convention of phase retrieval, measurement vectors corresponding to $A$ would be $a_i^*$ where $a_i, i = 1,2,\ldots m$ are the rows of $A$.  

Since observables are generically admissible, and the informational completeness of a set of admissible observables depends only on their eigenvectors, we from now on identify two admissible observables if their eigenspaces are the same, upon ordering by eigenvalue. Thus, an admissible observable $A \in \mathbb{H}_{n }$ with eigenvalues $\lambda_1 > \lambda_2,\ldots, >\lambda_n$ is identified with the equivalence class $\{ U \in \mathbb{U}_n; A = U^* D U, D_{ii} = \lambda_i \}$, where note that the eigenvectors of $A$ are conjugates of rows of $U$. Given a set of eigenvectors $u_1,\ldots u_n$ for $A$, we can explicitly represent this equivalence class as the set of matrices in $\mathbb{U}_n$, with $j$'th row equal to $e^{i \theta_j} u_j^*$, for some $\theta_j \in \mathbb{R}$.

\subsection{Connections to prior work}
There has been some recent progress in the study of phase retrieval, which bears on questions in quantum information completeness. Balan, Bodmann, Casazza and Edidin showed in \cite{BCE06,BCE07,BBC07,BBC09} that 4n-2 generic vectors in $\mathbb{C}^n$ are injective modulo phase, Hammen and Bodmann gave an example of $4n-4$ specific vectors with the same property \cite{Bodmann} and Hering, Vinzant, Conca and Edidin proved that 4n-4 generic vectors in $\mathbb{C}^n$ are injective modulo phase \cite{CEHV} .   

The basic strategy to establishing claims of the papers referenced in the previous paragraph is to express the set of measurement vectors which do not satisfy the injectivity property as a semi-algebraic set, and establish that this set has smaller algebraic dimension than the set of all measurement vectors. We employ a similar strategy in that we express the set of 4-tuples of observables which do not determine every pure state as a real algebraic variety $\mathcal{F} \subseteq \mathbb{U}_n^4$ and aim to show that it has smaller algebraic dimension than $\dim(\mathbb{U}_n^4)$. However, note that a rank-$n$ observable corresponds to a collection of $n$ rank-1 observables $\iff$ the n rank-1 observables form an orthonormal set. Thus, the analogy stops there, as previous algebro-geometric approaches to measurement injectivity relied on evident algebraic independence between the defining equations of analogously defined varieties \cite{BCE06,BCE07,BBC07,BBC09}. In contrast, in our unitary setting, the main difficulty is that there are non-trivial algebraic relations between the polynomials defining the unitary group and those that define $\mathcal{F}$ - thus simple dimension counts do not apply. 

The context of the aforementioned papers is therefore very different from ours and to the best of our knowledge, the main result in this paper is the first of its kind. Indeed, the existence of 4 unitary bases that are informationally complete, which follows from our main theorem, was listed as an open problem in the frame theory community \cite{RPRBP}, and was known to imply that phase retrieval from magnitudes of projections on subspaces of arbitrary dimensions is possible, a problem that occurs in crystal twinning during the process of X-ray crystallography \cite{RPRBP}.



On the algorithmic side, the authors of \cite{CSV} proposed an efficient algorithm for phase retrieval called PhaseLift, and proved that PhaseLift recovers pure states from quadratic measurements with $O(n \log n)$ iid gaussian vectors. Candes and Li later improved this result in \cite{Five} to requiring only $O(n)$ iid gaussian vectors. In comparison, the algorithmic result in our companion paper \cite{QTFFRO} requires a more technical approach, since there we work with full-rank observables, each yielding $n$ real numbers as measurements, which enforces unitary structure and thus a lack of independence between individual measurements.

\subsection{Outline of proof strategy}
We start by identifying 4-tuples of admissible observables with the Lie group of 4-tuples of unitary matrices $\mathbb{U}_{n}^4$ (thereby neglecting the specific values of the eigenvalues of each observable). We note that the set of 4-tuples of observables which do not determine every pure state is a real algebraic variety $\mathcal{F} \subseteq \mathbb{U}_n^4$ and our strategy is to show that it has smaller algebraic dimension than $\dim(\mathbb{U}_n^4)$. By using the invariance of $\mathcal{F}$ under coordinate-wise right-multiplication by an element of the unitary group $\mathbb{U}_n$, it is enough to show that $\mathcal{F}/\mathbb{U}_n$ has measure zero in $\mathbb{U}_n^4/\mathbb{U}_n \cong \mathbb{U}_n^3$. For a smooth manifold $M \subset \mathbb{U}_n^4$, if there exists a subgroup $G$ of $\mathbb{U}_n$, such that $M$ is invariant under multiplication by $G$ and $\dim(M) -\dim (G) < \dim (\mathbb{U}_n^3)$, then by factoring $M$ through by $G$ on the way to the total quotient, we could express $M/\mathbb{U}_n$ as the image of a smooth map from a manifold of smaller dimension than that of the target manifold, which by Sard's theorem has to be of measure zero. While the variety $\mathcal{F}$ is not a manifold, we seek a decomposition of it into appropriate smooth manifolds to which we can apply this argument. 

Again using the symmetry of $\mathbb{U}_n^4$, we show that a nicer un-twisted variety $\mathcal{F}' \subset \mathcal{F}$ has the same image as $\mathcal{F}$ under the quotient to $\mathbb{U}_n^3$. The Nash stratification theorem states that any real variety of algebraic dimension $d$ can be expressed as a finite union of smooth manifolds of geometric dimension at most $d$. We proceed by adapting the Nash stratification theorem to a group action, which yields decompositions of semialgebraic sets invariant under some Lie group action into smooth manifolds of the appropriate geometric dimension, which are also invariant under that action. To upper-bound the algebraic dimension of intermediary varieties arising in our case, we utilize a particular fiber-bundle structure of $\mathbb{U}_n^4$ and perform algebraic dimension calculations locally in charts by applying resolutions of singularities and working directly with the Zariski tangent spaces. We are thus able to cut up $\mathcal{F}'$ into a finite union of smooth manifolds adapted to an action by a large enough subgroup of the unitary group, which we use to factor through on the way to the total quotient by $\mathbb{U}_n$. Since under that intermediary quotient, the images of each smooth manifold composing $\mathcal{F}'$ have dimension less than $\dim{\mathbb{U}_n^3}$, this yields by Sard's theorem that the image of $\mathcal{F}$ under quotient by $\mathbb{U}_n$ has measure zero, which implies the main result.

\section{Statement of main results}
Having fixed $\mathbb{U}_n$ to be the $n\times n$ unitary group, consider $ \mathbb{U}_n^4 \subseteq \mathbb{C}^{4n\times n}$ as a product Lie group, with a group action given by coordinate-wise right multiplication: $\{U_i\}_{i=1}^4 \cdot U = \{U_i \cdot U\}_{i=1}^4$, for $U \in \mathbb{U}_n$ and $\{U_i\}_{i=1}^4 \in \mathbb{U}_n^4$. For an element $A$ of $\cS$, we denote its orbit by this action as $A \cdot \mathbb{U}_n = \{A\cdot U;  \quad U \in \mathbb{U}_n \}$. Letting $\mathbb{H}_n$ be the set of Hermitian $n \times n$ matrices, we will identify any $A \in \mathbb{C}^{m\times n}$ with a linear map
 \[
 \cA_A: X \in \mathbb{H}_n \mapsto  \{\tr(\bar{z}_i \bar{z}_i^* X)\}_{i=1}^{m} \in\mathbb{R}^m
 \]
  where $z_i$ are the rows of $A$. Note that for $X = xx^*$, and $\{e_i\}_{i=1}^m \in \mathbb{R}^m$ being the standard basis, we have
  \[
  \cA_A(xx^*) = \{\tr(\bar{z}_i \bar{z}_i^* xx^*)\}_{i=1}^{m} = \{ |e_i^*Ax|^2 \}_{i=1}^m \equiv |Ax|^2 \in \mathbb{R}^m
  \]
  

Viewing $A \in \cS$ as an element of $\mathbb{C}^{4n \times n}$, recall that $A$ injective modulo phase if for any $x,y \in \mathbb{C}^n$,
\[
|Ax|^2 = \cA_A(xx^*) = |Ay|^2 = \cA_A(yy^*) \implies xx^* = yy^*
\]

Note that the previous statement is equivalent to $\mathcal{A}_A$ being injective over rank-1 psd matrices. Since for any $A \in \cS \subseteq \mathbb{C}^{4n \times n}$, we have $ \mathcal{A}_A(xx^*) = \left| A x \right|^2 \in \mathbb{R}^{4n}$, $A \in \cS$ is injective modulo phase if and only if every element of the orbit $A \cdot \mathbb{U}_n$ is injective modulo phase. Therefore, the set of elements of $\cS$ that are not injective modulo phase, is invariant under the action of $\mathbb{U}_n$. We can now state the main theorem:



\begin{theorem}
\label{theorem: main}
Consider $\mathbb{U}_n$, for $n\geq 1$, acting on $\cS$ by coordinate-wise right multiplication and let $\pi_1$ be the quotient map of this action. Let $\mathcal{F}$ be the set of elements of $\cS \subseteq \mathbb{C}^{4n \times n}$ that are not injective mod phase. That is, let
\[
\mathcal{F} = \{A \in \cS; \quad \exists  x, y \in \mathbb{C}^n, \quad  xx^* \neq yy^*, \quad \cA_A(xx^*) = \cA_A(yy^*)\}
\]
Then, $\pi_1(\mathcal{F})$ is a set of measure zero in $\pi_1(\cS) = \cS/\mathbb{U}_n \cong \mathbb{U}_{n}^3$, with respect to Haar measure on $\mathbb{U}_{n}^3$. In particular, this implies that almost every quadruple of observables is informationally complete. 
\end{theorem}
To see the last implication, note that since $\mathcal{F}$ is invariant under the action of $\mathbb{U}_n$ and $\pi_1(\mathcal{F})$ has measure zero, we have that $\mathcal{F}$ is a set of measure zero in $\mathbb{U}_n^4$. 
Now, recall that we identified each admissible observable with its set of possible orthonormal eigenbases, sorted by decreasing eigenvalue. For a quadruple of observables, with orthonormal eigenbases given by $U_1,\ldots U_4 \in \mathbb{U}_n$, consider $A = (U_1^*,\ldots, U_4^*) \in \mathbb{U}_n^4\subset \mathbb{C}^{4n\times n}$. Thus, a quadruple of admissible observables, modulo their eigenvalues, is identified with the set of matrices $\{ B \in \mathbb{U}_n^4 \subseteq \mathbb{C}^{4n \times n}; \quad e_j^* B = e^{i \theta_j} e_j^* A, \quad \theta_j \in \mathbb{R}, \quad j = 1,2\ldots 4n \}$, which is also the orbit of $A$ under the action $g_\theta$ of multiplication of each row of $A$ by an element of a distinct copy of $\mathbb{S}^1$. Let $\pi_{\theta}$ be the quotient map associated with this action on $\mathbb{U}_n^4$. Noting that non-admissible observables are a set of measure zero in $\mathbb{H}_n$, to conclude the desired statement we must show that $\pi_{\theta}(\mathcal{F})$ has measure zero in $\pi_\theta (\mathbb{U}_n^4)$, which follows from the fact that $\mathcal{F}$ is $g_\theta$-invariant.

Thus distribution measurements with almost every quadruple of observables determines any pure state. Moreover, this result is sharp in that for $n\geq 6$, at least 4 observables are required to determine any pure state, which follows from results of \cite{QTPI}.

\begin{corollary}
\label{corollary:main}
Almost every element of $\mathbb{U}_n^4$ with $n\geq 1$ yields an embedding of $\mathbb{C}P^{n-1}$ into $\mathbb{R}^{4n-4}$. That is
when $n\geq 1$, for almost every element $A = \{ U_j \}_{j=1}^4 \in \cS$, where $u_{i}^j$ is the $j$'th column of $U_i$, the map
\[
x \in \mathbb{C}P^{n-1} \mapsto  \left\{ \frac{1}{\|x\|_2^2}  \left(\left| \left< u_{i}^{(1)}, x\right> \right|^2,\left| \left< u_{i}^{(2)}, x\right> \right|^2,\left| \left< u_{i}^{(3)}, x\right> \right|^2,\left| \left< u_{i}^{(4)}, x\right> \right|^2 \right) \right\}_{i=1}^{n-1} \in \mathbb{R}^{4n-4}
\]
is an embedding of $\mathbb{C}P^{n-1}$ into $\mathbb{R}^{4n-4}$.
 \end{corollary}

The corollary follows by first applying Theorem \ref{theorem: main} to $\{ U_i^* \}_{i=1}^4$ and using the results of \cite{QTPI, SFM}, in which it is shown that each $A$ in $\mathbb{C}^{m\times n}$ gives a smooth map from $\mathbb{C}P^{n-1}$ into $\mathbb{R}^{m}$ via
\[
x \in \mathbb{C}P^{n-1} \mapsto \frac{xx^*}{\|x\|_2^2} \in \mathbb{H}_n \mapsto \frac{1}{\|x\|_2^2}\mathcal{A}_A(xx^*) \in \mathbb{R}^{m}
\]
and that this map is an embedding $\iff$ $\cA_A$ is injective on rank-1 psd matrices, which is equivalent to injectivity of $A$ modulo phase. Thus, almost every $\{ U_i^* \}_{i=1}^4 \in \cS \in \mathbb{C}^{4n \times n}$ is injective modulo phase.

Now, take any $A = ( U_1, U_2, U_3, U_4 ) \in \mathbb{U}_n^4 \subseteq \mathbb{C}^{4n \times n}$ that is injective modulo phase. Let $A_1 \in \mathbb{C}^{(4n - 4) \times n}$ consist of all the rows of $A$ except those corresponding to the last row from each $U_i$. Since $\|Ux\|_2^2 = \|x\|_2^2$ for any unitary matrix $U$ and $x \in \mathbb{C}^n$, knowledge of $\|x\|_2^2$ allows us to throw away the last measurement from each unitary matrix without losing information. Thus, we have that the map from $\mathbb{C}P^{n-1}$ to $\mathbb{R}^{4n-4}$ induced by $A_1$ is an embedding, establishing Corollary \ref{corollary:main}. 

Note that $\mathbb{C}P^n$ does not embed into $\mathbb{R}^m$ for $m \leq 4n-\alpha(n)$, where $\alpha$ is the number of $1$'s in the binary expansion of n \cite{IEPS}. Thus, the embedding dimension in Corollary \ref{corollary:main} is optimal for $n = 2^k+1$ for any positive integer $k$, and is otherwise off from optimal by at most a logarithmic factor in $n$.

  

Similarly, having fixed $A = \{U_i\}_{i=1}^4 \in\mathbb{U}_n^4$ to be injective modulo phase, let $A_2 \in \mathbb{C}^{(4n - 3) \times n}$ consist of all the rows of $A$ except those corresponding to the last row of $U_2,U_3, U_4$. Since $\|U_1 x\|_2$ determines $\|x\|_2$, we then have that $A_2$ is injective modulo phase. Therefore, measurements $|A_2 x|^2 \in \mathbb{R}^{4n-3}$ determine any $x$ as an element of  
$\mathbb{C}^n/\mathbb{S}^1$.



\subsection{Exact pure-state recovery via PhaseLift}

It was proven in \cite{CSV} that the PhaseLift algorithm recovers signals $x \in \mathbb{C}^n$ exactly from $m = O(n \log n)$ measurements $\{\left| \left< x, z_i \right> \right|^2\}_{i=1}^m$ with high probability when the measurement vectors $z_i \in \mathbb{C}^n$ are iid gaussian and that this procedure is provably stable with respect to measurement noise under the same assumptions. To be precise, this means that in the noiseless case, for a fixed $x \in \mathbb{C}^n$ and defining the linear operator $\cA: X \in \mathbb{C}^{n \times n} \mapsto \{\tr(X z_i z_i^*) \}_{i=1}^m$, the program
\begin{equation}
\label{eq:tracemin}
 \begin{array}{ll}
    \text{minimize}   & \quad \tr(X)\\ 
    \text{subject to} & \quad  \cA(X) = \cA(xx^*)\\
& \quad X \succeq 0;  
\end{array}
\end{equation}
recovers $xx^*$ with high probability. The stability result uses a modified, noise-aware convex program. These guarantees were subsequently improved to hold uniformly over all signals for $m = O(n)$ with sharp stability guarantees in \cite{Five} and it was shown in \cite{SOR} that in the noiseless case, this program has only one point in its feasible set, namely $xx^*$. 


In our setting of recovery from measurements with full-rank observables, we have $z_i$ as columns of iid Haar distributed unitary matrices. It was proven in a companion paper \cite{QTFFRO} by the latter present author that PhaseLift succeeds with high probability under this unitary measurement model as long as the number of unitary matrices used is $O(1)$ (which corresponds to $m$ = O(n) in the above setting). Specifically, the main theorem from \cite{QTFFRO} reads:


\begin{theorem}
Take $x \in \mathbb{C}^n$ and assume that measurements of the form $\{|U_{k}^*x|^2\}_{k=1}^r$ are available, where the $U_i$ are sampled independently according to the Haar measure on $\mathbb{U}_n$, the unitary group, so that the total number of measurements is $m = rn$. Then the PhaseLift algorithm succeeds in recovering $x$ up to global phase with very high probability with $r = \text{O}(1)$. 
\end{theorem}

Thus, informational completeness is off from efficient recovery by a constant oversampling factor. 

\section{Proof of the main result}
We begin with some simplifying lemmas. Lemma 9 in \cite{SavingPhase} is similar in spirit, but a stronger statement holds in the unitary setting: 
\begin{lemma}
Let $A \in \mathbb{U}_n^4 \subset \mathbb{C}^{4n \times n}$ and call $\cA = \cA_A$. Then $A$ is not injective mod phase $\iff$ there is a rank-2 Hermitian matrix with eigenvalues $1,-1$ in the nullspace of $\cA$. 
\end{lemma}
\begin{proof}
First, take any rank-2 indefinite matrix $X \neq 0$. It can be written as $X = xx^* - yy^*$ for some non-zero $x,y \in \mathbb{C}^n$. If $\cA(X) = 0$, then $\cA(xx^*) = \cA(yy^*)$ and thus $A$ is not injective modulo phase.

Now, assume that $A$ is not injective modulo phase. Thus $\cA(xx^*) = \cA(yy^*)$ for some $x,y \in \mathbb{C}^n$ such that $xx^* \neq yy^*$. Defining $X = xx^* - yy^*$, this gives $\cA(X) = 0$. We have that necessarily $xx^* \neq 0$ and $yy^* \neq 0$ because if, say wlog $xx^* = 0$, then $\cA(xx^*) = 0 \implies \cA(yy^*) = 0$, but since 
\[
0=\|\cA(yy^*)\|_1 = \sum_{i=1}^m \tr(yy^* \bar{z}_i \bar{z}_i^*) = \sum_{i=1}^m \left| \left< \bar{z}_i, y \right> \right|^2 = 4\|y\|_2^2
\]
, where $z_i$ are the rows of $A$, this implies that $y = 0$, which contradicts $xx^* \neq yy^*$. Thus $X$ is an indefinite Hermitian matrix. By linearity, we can assume that $\|X\|_F = \sqrt{2}$, where $\|.\|_F$ is the Frobenius norm. Now, consider the eigenvalue decomposition of $X = xx^* - yy^*$: 
\[
X = \lambda_1 uu^* + \lambda_2 vv^*
\]
with eigenvalues $\lambda_1>0, \lambda_2 < 0$. where $\left<u,v\right> = 0$ and $\|u\|_2 = \|v\|_2 = 1$. Then since $\cA(X) = 0$, we have $\lambda_1 \cA(uu^*) = -\lambda_2 \cA(vv^*)$ and since $\cA(uu^*) \geq 0$, we have
 \[
\lambda_1 \|\cA(uu^*)\|_1 = -\lambda_2 \|\cA(vv^*)\|_1  \implies \lambda_1 = -\lambda_2
\]
 since $\|\cA(uu^*)\|_1 = \|\cA(vv^*)\|_1 = 4$. By $\|X\|_F^2 = 2 =  \lambda_1^2 + \lambda_2^2$, we have $\lambda_1 = 1, \lambda_2 = -1$. Thus, if $\cA$ is not injective modulo phase, there exists a rank 2 indefinite Hermitian matrix in the nullspace of $\cA$, with eigenvalues 1, -1.
\end{proof}

Recall that
\[
\mathcal{F} = \{A \in \cS; \quad \exists  x, y \in \mathbb{C}^n, \quad  xx^* \neq yy^*, \quad \cA_A(xx^*) = \cA_A(yy^*)\}
\]
Now, we will define a simpler variety inside $\mathcal{F}$, which generates $\mathcal{F}$ under coordinate-wise right-multiplication by $\mathbb{U}_n$. Letting $e_i\in\mathbb{C}^n$ denote the standard basis vectors, define the set 
\[
N_{e_1,e_2} = \{A\in \cS; \quad \cA_{A} (e_1 e_1^* - e_2 e_2^*) = 0 \}
\]
\begin{lemma}
Let $\pi_1$ denote the quotient map associated to the action of $\mathbb{U}_n$ by coordinate-wise right multiplication on $\cS$. Then $\pi_1(\mathcal{F}) = \pi_{1}(N_{e_1,e_2})$.
\end{lemma}
\begin{proof}
Note that $N_{e_1,e_2} \subseteq \mathcal{F}$. Assume that some $\cA_A$, corresponding to $A \in \cS$, is not injective mod phase. By Lemma 2.3, we must have $\cA_A(xx^* - yy^*) = 0$ for some unit normed and orthogonal $x,y \in \mathbb{C}^n$. Now, take some $U \in \cU$ such that $U e_1 = x, U e_2 = y$. Then, $\cA_{A\cdot U}$ satisfies
 \[
  \cA_{A\cdot U}(e_1 e_1^* - e_2 e_2^*) = \left| AU e_1 \right|^2 - \left| A U e_2 \right|^2 = \left| Ax \right|^2 - \left| A y \right|^2 = \cA_{A}(xx^*-yy^*) = 0.
  \]
Since $\pi_1(A\cdot U) = \pi_1(A)$, we have that 
\[
A \in \mathcal{F} \implies N_{e_1,e_1} \bigcap \left(A\cdot\cU\right) \neq \emptyset.
\]
This, coupled with $N_{e_1,e_2} \subseteq \mathcal{F}$, implies that $\pi_1(\mathcal{F}) = \pi_1 (N_{e_1,e_2})$. 
\end{proof}

The point of this lemma is that since $\pi_1(\mathcal{F}) = \pi_1 (N_{e_1,e_2})$, it suffices to show that $\pi_1(N_{e_1,e_2})$ has measure zero in $\cU^4/\mathbb{U}_{n}$ to establish the main theorem.

\begin{lemma}
\label{lemma: Qdim}
Let $M$ be a smooth manifold and let $G, G'$ be compact Lie groups which act smoothly, freely and properly on $M$, such that $G' \leq G$. Assume that $(M, B, \pi, F)$ is a fiber bundle with projection map $\pi$, base space $B$ and fiber F, such that $\pi(pG') = \pi(p)$ for any $p \in M$. Now, let $N'$ be a submanifold of $B$. Then, $N=\pi^{-1}(N')$ is a $G'$-stable submanifold of $M$, with 
\[
\dim(N) = \dim(N') + \dim(F),
\]
and if 
\[
\dim(N/G') < \dim(M/G),
\]
we have that $N/G$ has measure zero in any chart on $M/G$. In particular, if $M/G$ is a Lie group, $N/G$ has measure zero with respect to the Haar measure on $M/G$.
\end{lemma}
\begin{proof}
By properties assumed of $G$ and $M$, $M/G$ is a smooth manifold and the quotient map
\[
\pi_1 : M \mapsto M/G
\]
is smooth. Moreover, 
\[
\dim(M/G) = \dim(M) - \dim(G).
\]
Since $M$ is a fiber bundle, we have that for any point $p \in M$, there is a chart 
\[
\left(U \times Y \subseteq M, \phi: p \in U \times Y \mapsto (u_1,u_2,\ldots u_l, x_1, \ldots x_n) \in U'\times Y' \subseteq \mathbb{R}^{\dim(B)}\times \mathbb{R}^{\dim(F)})\right)
\]
where $l = \dim(B)$, $U$ is an open neighborhood of $\pi(p)$ and $U'$ and $Y'$ are open subsets of $\mathbb{R}^{\dim(B)}$ and $\mathbb{R}^{\dim(F)}$. Now, since $N'$ is a submanifold of the base space $B$, and $\left(U, u_1,\ldots u_l\right)$ is a chart for $\pi(p) \in B$, we can refine the coordinates $u_1,\ldots, u_l$ such that the submanifold $N'$ can be expressed locally as $(u_1=0,\ldots u_r=0, u_{r+1},\ldots u_{l})$, where $r = \dim(B) - \dim(N')$. Therefore, 
\[
(u_1=0,\ldots u_r = 0, u_{r+1},\ldots u_l, x_1,\ldots x_n)
\]
gives coordinates for $N$ as a submanifold of $M$. Note that $N$ is $G'$-stable. The action of $G'$ on $M$ restricts to a smooth and free action on $N$, which is furthermore proper since $G'$ is compact. We then have that the associated quotient map
\[
\pi_2 : N \mapsto N/G'
\]
is a surjective submersion and $N/G'$ is a smooth manifold, with $\dim(N/G') = \dim(N)-\dim(G')$. Since $N$ is a submanifold, $\pi_1$ restricts to a smooth map on $N$ and thus, since $\pi_1 \vert_{N} = g \circ \pi_2$, where
\[
g: N/G' \mapsto M/G
\]
sends an element of $N/G'$ to its $G$-orbit in $M/G$, we have that $g$ is smooth by Proposition 5.19 in \cite{LSM}. By construction,
\[
N/G = \pi_1(N) = g\circ \pi_2 (N) = g(N/G') \subseteq M/G.
\]
Thus, $N/G$ is the image of a smooth map, in a manifold of dimension  $\dim(M) - \dim(G)$, from a manifold of dimension $ \dim(N) - \dim(G') < \dim(M) - \dim(G)$. By Sard's theorem, we have therefore that $N/G$ has measure zero in $M/G$.

\end{proof}

Using the notation of Lemma \ref{lemma: Qdim}, let $M = \mathbb{U}_{n}^4$, $G = \mathbb{U}_n$ and 
\[
G' =\{U \in \cU; \left|e_i^*Ue_i\right| = 1, i = 1,2\} \cong \mathbb{S}^1 \times \mathbb{S}^1 \times \mathbb{U}_{n-2}.
\]
Consider $G,G'$ acting by right multiplication on $M$. We will show that the set $N_{e_1,e_2}$ can be expressed as a union of manifolds which satisfy the properties of $N$ in Lemma \ref{lemma: Qdim}. 

First note that $G$ and $G'$ both act smoothly, freely and properly by right multiplication on $M$, the last property due to each being a compact Lie group. Define
\[
G'' = \{ \{U_{i}\}_{i=1}^4 \in \cS; \left|e_1^*U_{i}e_1\right| =  \left|e_2^*U_{i}e_2\right| = 1, i = 1,\ldots,4, e_j^*U_{l}e_j = e_j^*U_{k}e_j, j=1,2. \quad 1\leq l < k \leq 4 \} 
\]
as a subgroup of the product Lie group $\mathbb{U}_{n}^4$ and let $G''$ act on $\cS$ by right multiplication in each component. Since $G''$ is a closed Lie subgroup of $\cU^4$, we have that 
\[
(M = \mathbb{U}_{n}^4, B = \mathbb{V}_{2}(\mathbb{C}^n)^4 / (\mathbb{S}^1 \times \mathbb{S}^1), \pi, F = \mathbb{S}^1 \times \mathbb{S}^1 \times \mathbb{U}_{n-2}^4 \cong G'')
\]
is a fiber bundle, with base space
\[
B = \cS/G'' \cong \mathbb{V}_{2}(\mathbb{C}^n)^4 / (\mathbb{S}^1 \times \mathbb{S}^1) ,
\]
where $\mathbb{V}_{2}(\mathbb{C}^n)$ is the Stiefel manifold of complex orthonormal 2-frames, the projection map
$\pi$ is the quotient map associated to the action of $G''$ and the fiber $F$ is diffeomorphic to $G'' 
\cong \mathbb{S}^1 \times \mathbb{S}^1 \times \mathbb{U}_{n-2}^4$. The quotient by $\mathbb{S}^1 \times \mathbb{S}^1$ is to be interpreted as given by the equivalence relation 
\[
 (u_1^1,\ldots u_1^4, u_2^1,\ldots u_2^4) \equiv (e^{i\theta_1}u_1^1,\ldots e^{i\theta_1}u_1^4, e^{i\theta_2}u_2^1,\ldots e^{i\theta_2}u_2^4)
 \]
 for any $\theta_i \in \mathbb{R}$.

It is clear that for any $p \in \mathbb{U}_{n}^4$, we have $\pi(pG') = \pi(p)$, since $G'$ can be thought of as a subgroup of $G''$ in the product Lie group $\mathbb{U}_{n}^4$. Moreover, $M/G \cong \mathbb{U}_{n}^3$ is a compact Lie group. Thus, $G,G'$ and $(M,B,\pi,F)$ satisfy the conditions of Lemma \ref{lemma: Qdim}.

Consider $P = \pi(N_{e_1,e_2}) = N_{e_1,e_2}/G''$ as a subset of the base space. We state here an intermediary theorem which we prove in the next section:

\begin{theorem}
\label{theorem: dimension}
$P$ may be expressed as
\[
P = \bigcup_{\alpha = 1}^k P_\alpha ' \subseteq \mathbb{V}_{2}(\mathbb{C}^n)^4 / (\mathbb{S}^1 \times \mathbb{S}^1),
\]
for some integer $k$, where each $P_{\alpha}'$ is a submanifold of $\mathbb{V}_{2}(\mathbb{C}^n)^4 / (\mathbb{S}^1 \times \mathbb{S}^1)$ and $\dim(P_\alpha ') \leq 4(3n-3)-2$. 
\end{theorem}
Thus, using this theorem, Lemma \ref{lemma: Qdim}, and noting that $N_{e_1,e_2}$ is $G''$-stable, we have that
\[
N_{e_1,e_2} = \pi^{-1}(P) = \pi^{-1}(\bigcup_{\alpha=1}^k P_\alpha ') = \bigcup_{\alpha=1}^k \pi^{-1}(P_\alpha ')
\]
 is itself a union of submanifolds: $N_{e_1,e_2} = \bigcup_{\alpha=1}^k P_\alpha$, where $P_\alpha = \pi^{-1}(P_\alpha ')$ and furthermore,
 \[
 \dim(P_\alpha) \leq \dim(P_\alpha ') + \dim(F) =  \dim(P_\alpha') + \dim(\mathbb{S}^1 \times \mathbb{S}^1 \times \mathbb{U}_{n-2}^4) \leq 4(3n-3) + 4(n-2)^2 
 \]
 Also, each $P_\alpha$ is $G'$-stable, and modding out by $G'$ we have
 \[
 \dim(P_\alpha/G') = \dim(P_\alpha) - \dim(G') = 3n^2-2. 
 \]
 Thus, since $\dim(P_\alpha/G') \leq 3n^2-2 < 3n^2 =\dim(M/G)$, Lemma \ref{lemma: Qdim} gives that each $P_\alpha/G$ has measure zero in $M/G$. 
 
 Now, since
 \[
 \pi_1(N_{e_1,e_2}) = \pi_1(\bigcup_{\alpha=1}^k P_\alpha) = \bigcup_{\alpha=1}^k \pi_1(P_\alpha),
 \]
 we have that $N_{e_1,e_2}/G$ has measure zero in $\cU^4/G$, because the union is finite. This implies that $\mathcal{F}/\mathbb{U}_{n}$ has measure zero in $\cS/\mathbb{U}_n$, completing the proof of Theorem \ref{theorem: main}.
  
\subsection{Proof of Theorem \ref{theorem: dimension}}
Define the space $Z=(\mathbb{C}^{2n})^4/\mathbb{S}^1 \times \mathbb{S}^1 = (\mathbb{R}^{4n})^4/\mathbb{S}^1 \times \mathbb{S}^1$ and consider 
\[
P = N_{e_1,e_2}/(\mathbb{S}^1 \times \mathbb{S}^1 \times \mathbb{U}_{n-2}^r) \subseteq B =\St^4/(\mathbb{S}^1 \times \mathbb{S}^1) \subseteq Z. 
\]
We have $P$ =
\[
\{\{(u_1^j,u_2^j)\}_{j=1}^4 \in (\mathbb{C}^{2n})^4; \|u_1^j\|_2 = \|u_2^j\|_2 = 1, \left<u_1^j,u_2^j \right> = 0, \left|u_{1i}^{j}\right|^2 = \left| u_{2i}^{j}\right|^2, i=1,2, \ldots n, j = 1,2,3, 4 \} / (\mathbb{S}^1 \times \mathbb{S}^1).
\]

For $1\leq i < j \leq n$, at a point $z \in Z$ for which $u_{1i}^1 \neq 0, u_{2j}^1 \neq 0$, consider the following charts on $Z$,
\[
\left(U_{ij} = \{\{(u_1^j,u_2^j)\}_{j=1}^4 \in Z; u_{1i}^1 \neq 0, u_{2j}^1 \neq 0 \}, \phi_{ij}\right)
\]
The coordinate maps $\phi_{ij}$ on these charts send 
\[
\{(u_1^j,u_2^j)\}_{j=1}^4 \in Z \mapsto  (\pi_{/ij}(e^{i\theta_{1}}u_1^j, e^{i\theta_{2}}u_2^j), \{(e^{i\theta_{1}}u_1^j, e^{i\theta_{2}}u_2^j)\}_{j = 2}^4) \in \mathbb{R}^{4n\times 4 - 2}
\]
where $\pi_{/ij}$ takes $ (u_{1}^1,u_2^1)\in \mathbb{C}^{2n}$ to $\mathbb{R}^{4n-2}$ by keeping all but the imaginary parts of $u_{1i}^1$ and $u_{2j}^1$ and $e^{i\theta_1}, e^{i\theta_2}$ are chosen such that $im(u_{1i}^1)=0$ and $im(u_{2j}^1)=0$.

By orthonormality of vectors in $\St$, we have that
\[
P = \bigcup_{1\leq i < j\leq n} P \cap U_{ij}
\]
Define
\[
W = \{ (u_1,u_2) \in \St; |u_{1i}| = |u_{2i}|, i = 1,2\ldots n\} \subseteq \mathbb{C}^{2n}
\]
where $\St$ is the Steifel manifold of two orthonormal complex n-dimensional vectors, and let
\[
W_{ij} = \left\{(u_1,u_2) \in W; u_{1i} \neq 0, u_{2j} \neq 0, im(u_{1i})=im(u_{2j})=0\right\} \subseteq \mathbb{C}^{2n}
\]
In coordinates on the charts $U_{ij}$, we have
\[
\phi_{ij}(P \cap U_{ij}) = \pi_{/ij}(W_{ij}) \times W^{3}
\]

\begin{lemma}
\label{lemma: dims}
$W$ and $W_{ij}$ are semialgebraic sets in $ \mathbb{R}^{4n}$, with $\dim(W) \leq 3n-3$ and  $\dim(W_{ij}) \leq 3n-5$.
\end{lemma}
By the Nash stratification theorem, Proposition 9.1.8 in \cite{RAG}, any semialgebraic set is a union of Nash submanifolds. Therefore, we can express any $P \cap U_{ij}$ as a union of Nash submanifolds of $U_{ij}$ and therefore $P$ is a union of submanifolds of $Z$. Now, since $B$ is a submanifold of $Z$ and $P \subseteq B$, P is also a union of submanifolds of $B$ (by submanifold we always mean embedded submanifold). The dimension of any of these submanifolds is clearly upper bounded by 
\[
\dim(\pi_{/ij}(W_{ij}) \times W^{3}) \leq 4(3n-3) -2
\]
Since $\pi_{/ij}$ cannot increase algebraic dimension, this completes the proof of theorem \ref{theorem: dimension}, once we prove the lemma \ref{lemma: dims} below.

\subsection{Proof of Lemma \ref{lemma: dims}}
\subsection{An Auxiliary Variety}

Let $I$ denote the ideal 
\[
(f_1,\ldots ,f_n ,g,h_1,h_2)\subset \R[x_1,\ldots, x_n,y_1,\ldots, y_n,v_1,\ldots,v_n, w_1,\ldots, w_n]
\]
where 
\[
f_i = x_i^2 + y_i^2 - 1,\quad g=\sum_{j=1}^n (v_j^2 + w_j^2) -1,\quad h_1 = \sum_{j=1}^n (v_j^2 + w_j^2)x_j,\quad h_2 = \sum_{j=1}^n (v_j^2 + w_j^2)y_j
\]

\begin{lemma} Let $X$ be the algebraic set in $\mathbb{R}^{4n}$ defined by the ideal $I$.  Then $X$ is a smooth scheme-theoretic complete intersection of dimension $3n-3$.
\label{lemma: aux}
\end{lemma}
\begin{proof}
There is a smooth action of $(\mathbb{S}^1)^n$ on $\mathbb{R}^{4n}$, defined by 
\[
(\mu , (v,w,x,y) ) \in (\mathbb{S}^1)^n \times \mathbb{R}^{4n} \mapsto (\mu \circ (v+iw),x+iy)
\] 
Since for each $\mu$, this map is a diffeomorphism, the Zariski tangent spaces of $X$ will be isomorphic along orbits of this action and thus we need only consider a representative of each orbit. In particular, we consider points where $w_i = 0$ for all $i$.

The tangent space at a point $p = (\tilde{v}_1, \ldots, \tilde{v}_n, \tilde{w}_1 = 0, \ldots, \tilde{w}_n = 0, \tilde{x}_1, \ldots, \tilde{x}_n, \tilde{y}_1, \ldots, \tilde{y}_n) \in X$ is the orthogonal complement of the subspace spanned by the following differentials:
\begin{align*}
& F_i = \tilde{x}_i \frac{\partial}{\partial x_i} + \tilde{y}_i \dy, \quad i = 1,2,\ldots, n\\
& G = \sum_{i=1}^n 2\tilde{v}_i \dv\\
& H_1 = \sum_{i=1}^n 2\tilde{v}_i\tilde{x}_i \dv + \tilde{v}_i^2 \dx\\
& H_2 = \sum_{i=1}^n 2\tilde{v}_i\tilde{y}_i \dv + \tilde{v}_i^2 \dy
\end{align*}
We will show that these differentials are linearly independent at every point $p \in X$, thereby establishing that $\dim(T_p (X)) = 4n- (3n+3) = 3n-3$. 

Suppose $(\sum_{i=1}^n a_i F_i) + bG + cH_1 + dH_2 = 0$ for some $(a_1,\ldots,a_n,b,c,d) \in \R^{n+3}$.  Then, collecting terms we get the following:
\begin{align*}
& (\dx) \quad \alpha_i := a_i\tilde{x}_i + c\tilde{v}_i^2 = 0\\
& (\dy) \quad \beta_i := a_i\tilde{y}_i + d\tilde{v}_i^2 = 0\\
& (\dv) \quad \gamma_i := \tilde{v}_i(b+c\tilde{x}_i + d\tilde{y}_i) = 0
\end{align*}

Define 
\[
\tilde{f}_1 = f_1(p), \ldots, \tilde{f}_n = f_n(p),\quad \tilde{g} = g(p), \quad \tilde{h}_1 = h_1(p),\quad \tilde{h}_2 = h_2(p)
\]
Since $p \in X$, $\tilde{f}_i = \tilde{g} = \tilde{h}_1 = \tilde{h}_2 =  0$ and in particular,
\begin{align*}
0 & = b\tilde{g} + c\tilde{h}_1 + d\tilde{h}_2 \\
& = b(\sum_{i=1}^n \tilde{v}_i^2 -1) + c\sum_{i=1}^n \tilde{x}_i \tilde{v}_i^2 + d\sum_{i=1}^n \tilde{y}_i \tilde{v}_i^2 \\
& = \left(\sum_i \tilde{v}_i (\tilde{v_i} (b+c\tilde{x}_i + d\tilde{y}_i))\right) - b \\
& = \left(\sum_i \tilde{v}_i (\gamma_i)\right) - b = -b
\end{align*}
Thus $b=0$.

Note that since $\tilde{g} = 0$, not all $\tilde{v}_i = 0$.  Say, without loss of generality, that $\tilde{v}_1 \neq 0$.  Then since $\gamma_1 = 0$,
$$b+c\tilde{x}_1 + d\tilde{y}_1 = c\tilde{x}_1 + d\tilde{y}_1 = 0$$

Continuing,
\begin{align*}
0 & =c\alpha_1 + d\beta_1 \\
& = c(a_1\tilde{x}_1 + c \tilde{v}_1^2) + d(a_1\tilde{y}_1 + d \tilde{v}_1^2)  \\
& =  a_1 (c\tilde{x}_1 + d\tilde{y}_1) + (c^2 + d^2)\tilde{v}_1^2 \\
& = (c^2 + d^2)\tilde{v}_1^2
\end{align*}
Since $\tilde{v}_1^2 \neq 0$, $c^2 + d^2 = 0$ and hence $c=d=0$ as $c,d \in \R$.

But then, $\tilde{x}_i \alpha_i + \tilde{y}_i \beta_i = a_i (\tilde{x}_i^2 + \tilde{y}_i^2) = 0$.  As $\tilde{f}_i = 0$, we have $a_i (\tilde{x}_i^2 + \tilde{y}_i^2) = a_i$ and thus $a_i = 0$ for all $i$.  We've thereby shown that $(a_i,b,c,d) = 0$ and hence $X$ is smooth of dimension $\dim T_p X = 4n-(n+3) = 3n-3$.
\end{proof}

Note that by pairing real coordinates into complex ones, $X$ may be written set-theoretically as:
\[
X = \{ (v+iw,x+iy) \in \R^{4n}; (x+iy)\circ(x-iy) = 1,\left<v+iw,v+iw\right>=1,\left<v+iw,(v+iw) \circ (x+iy)\right> = 0\}
\]
where $\circ$ denotes the Hadamard product.

For the following two corollaries, we will need the polynomial map:
\[
S: (v,w,x,y) \in \mathbb{R}^{4n} \mapsto (v+iw, (v+iw)\circ(x+iy)) \in \mathbb{R}^{4n}
\]
The equations defining $X$ in complex coordinates say that $S$ surjects $X$ onto $W$. Since $S$ is a semialgebraic map, we have $\dim(W) \leq \dim(X) = 3n-3$.  We've shown:

\begin{corollary} $\dim W \leq 3n-3$
\label{corollary: dim1}
\end{corollary}

We will have thus completed the proof of Lemma $\ref{lemma: dims}$ once we show:

\begin{corollary} $\dim W_{ij} \leq 3n-5$
\label{corollary: dim2}
\end{corollary}
\begin{proof}
Consider the linear map $\phi_j$ defined on $\mathbb{R}^{4n}$ as
\begin{align*}
(u_{11} &,\ldots,u_{1n}, u_{21},\ldots,u_{2n})  \\
 & \mapsto (u_{11},\ldots,u_{1(j-1)},u_{2j},u_{1(j+1)},\ldots,u_{1n},u_{21},\ldots,u_{2(j-1)},\bar{u}_{1j},u_{2(j+1)},\ldots,u_{2n})
\end{align*}
i.e. the identity on all components of $(u_1,u_2) \in \mathbb{C}^{2n}$ except the $j$-th ones where it sends 
\[
(re(u_{1j}),im(u_{1j}),re(u_{2j}),im(u_{2j})) \mapsto (re(u_{2j}),im(u_{2j}),re(u_{1j}),-im(u_{1j}))
\]
This map is semialgebraic and it is easy to verify that it is a bijection between $W_{ij}$ and
\[
W_{ij}' := \left\{(u_1,u_2) \in W;u_{1i}\neq 0,u_{1j}\neq 0,im(u_{1i})=im(u_{1j})=0\right\}.
\]
Therefore, it is enough to upper bound the dimension of $W_{ij}'$. Again, wlog, we take $i=1, j=2$.

Consider the subvariety
\begin{align*}
Y &:= \{ (v+iw,x+iy) \in X; w_1=w_2=0\} \\
   &= \{ (v+iw,x+iy) \in \R^{4n}; w_1 = w_2 = 0,|x+iy| = 1,\left<v+iw,v+iw\right>=1,\left<v+iw,(v+iw) \circ (x+iy)\right> = 0\} \\
   & \subseteq X
\end{align*}
This clearly surjects onto $\left\{(u_1,u_2) \in W;im(u_{1i})=im(u_{1j})=0\right\} \supseteq W_{ij}'$ via the map $S$ defined above.  Thus, it will be enough to show $\dim Y \leq 3n-5.$

$\mathbb{C}^{2n}$ admits an action of $(\mathbb{S}^1)^{n-2}$ on the final $n-2$ components of the first vector, which is just the restriction of the $(\mathbb{S}^1)^n$ action on $\mathbb{C}^{2n}$ to the subgroup $(\mathbb{S}^1)^{n-2} \hookrightarrow (\mathbb{S}^1)^n$ where $(\mu_3,\ldots,\mu_n) \mapsto (1,1,\mu_3,\ldots,\mu_n)$.  Tangent space dimensions are equal along orbits of this action and so, as above, we need only consider representatives where all $w_i = 0$.  At such a point $p = (\tilde{v}_1, \ldots, \tilde{v}_n, \tilde{w}_1 = 0, \ldots, \tilde{w}_n = 0, \tilde{x}_1, \ldots, \tilde{x}_n, \tilde{y}_1, \ldots, \tilde{y}_n) \in X$, the tangent space is the orthogonal complement (in the vector space spanned by $\frac{\partial}{\partial v_i},\frac{\partial}{\partial w_i},\frac{\partial}{\partial x_i},\frac{\partial}{\partial y_i}, i = 1,2\ldots, n)$ of the vectors 

\begin{align*}
&\frac{\partial}{\partial w_1},\quad \frac{\partial}{\partial w_2}\\
& F_i = \tilde{x}_i \frac{\partial}{\partial x_i} + \tilde{y}_i \dy, \quad i = 1,2,\ldots, n\\
& G = \sum_{i=1}^n 2\tilde{v}_i \dv\\
& H_1 = \sum_{i=1}^n 2\tilde{v}_i\tilde{x}_i \dv + \tilde{v}_i^2 \dx\\
& H_2 = \sum_{i=1}^n 2\tilde{v}_i\tilde{y}_i \dv + \tilde{v}_i^2 \dy
\end{align*}

That is,
\begin{align*}
T_p Y & = \R\left<\dw,\dv,\dx,\dy\right>/\R\left<\frac{\partial}{\partial w_1},\frac{\partial}{\partial w_2},F,G,H_1,H_2\right> \\
& \cong \R\left<\frac{\partial}{\partial w_3},\ldots,\frac{\partial}{\partial w_n}\right> \oplus (\R\left<\dv,\dx,\dy\right>/\R\left<F,G,H_1,H_2\right>)
\end{align*}

the last isomorphism since none of $F,G,H_1,H_2$ involve $w$'s.  Note that in the course of proving Lemma $\ref{lemma: aux}$ we actually showed $\dim\R\left<\dv,\dx,\dy\right>/\R\left<F,G,H_1,H_2\right> = 2n-3$ so that the tangent space of $Y$ has dimension $(n-2) + (2n-3) = 3n-5$ everywhere, implying $Y$ is smooth of dimension $3n-5$.
\end{proof}


%
%
%
%
%
%

%
%

%
%
%
%
%
%
%


\newcommand{\bx}{\bm{x}}
\newcommand{\bX}{\bm{X}}
\newcommand{\bH}{\bm{H}}
\newcommand{\bM}{\bm{M}}
\newcommand{\bI}{\bm{I}}
\newcommand{\bU}{\bm{U}}
\newcommand{\bY}{\bm{Y}}
\newcommand{\bZ}{\bm{Z}}
\newcommand{\bb}{\bm{b}}
\newcommand{\by}{\bm{y}}
\newcommand{\bu}{\bm{u}}
\newcommand{\bv}{\bm{v}}
\newcommand{\be}{\bm{e}}
\newcommand{\bz}{\bm{z}}
\newcommand{\bzi}{\bm{z_i}}

\newcommand{\ccS}{\mathcal{S}}

\section{Acknowledgements} We are particularly grateful to Terence Tao for pointing out that the probabilistic method as an approach to establishing minimal information completeness was unlikely to work, to Daniel Cristofaro-Gardiner for some very useful suggestions on Lemma 5.5, and to Harold Williams for early discussions on Wright's conjecture. We also acknowledge Emmanuel Candes, Bernd Sturmfels, Qingchun Ren, and Richard Schoen for discussions. Thanks to Dustin Mixon for maintaining his research blog "Short, Fat Matrices", via which we first became aware that Wright's conjecture is false and of the embedding obstructions that give lower bounds on the required number of measurements, and thanks to Thomas Strohmer for introducing us to Wright's conjecture. VV was supported by the Division of Mathematical Sciences, National Science Foundation, under Grant No. DMS-0913695.

\bibliographystyle{plain}
\bibliography{ucbtest}

\end{document}